\documentclass[submission,copyright,creativecommons]{eptcs}
\providecommand{\event}{Non-Classical Logics 2024 (NCL'24)} 
\usepackage{breakurl}             
\usepackage{underscore}           
\usepackage{latexsym}
\usepackage{amssymb}
\usepackage{amsmath,amsthm,bussproofs,units,enumerate}
\usepackage{graphicx}
\newcommand{\De}{\mathfrak{D}}
\newcommand{\Ex}{\mathcal{E}}
\newcommand{\I}{\mathtt{I}}
\newcommand{\Ix}{\mathtt{I}x}

\newcommand{\non}{\neg}
 \usepackage{multicol}
\theoremstyle{definition}
\newtheorem{defn}{\sc Definition}[section]
\newtheorem{thm}{\sc Theorem}[section]
\newtheorem{lemma}{\sc Lemma}[section]

\title{A Binary Quantifier for Definite Descriptions in Nelsonian Free Logic}
\author{
Yaroslav Petrukhin
\institute{Univeristy of Lodz\\
\L{}\'{o}d\'{z}, Poland}
\email{\quad yaroslav.petrukhin@gmail.com}
}
\def\titlerunning{A Binary Quantifier for Definite Descriptions in Nelson}
\def\authorrunning{Y. Petrukhin}
\begin{document}
\maketitle

\begin{abstract}
The method K\"{u}rbis used to formalise definite descriptions with a binary quantifier $ \I $, such that $ \I x[F,G]$ indicates `the $ F $ is $ G $', is examined and improved upon in this work. K\"{u}rbis first looked at $ \I $ in intuitionistic logic and its negative free form. It is well-known that intuitionistic reasoning approaches truth constructively. We also want to approach falsehood constructively, in Nelson's footsteps. Within the context of Nelson's paraconsistent logic \textbf{N4} and its negative free variant, we examine $ \I $. We offer an embedding function from Nelson's (free) logic into intuitionistic (free) logic, as well as a natural deduction system for Nelson's (free) logic supplied with $ \I $ and Kripke style semantics for it. Our method not only yields constructive falsehood, but also provides an alternate resolution to an issue pertaining to Russell's interpretation of definite descriptions. This comprehension might result in paradoxes. Free logic, which is often used to solve this issue, is insufficiently powerful to produce contradictions. Instead, we employ paraconsistent logic, which is made to function in the presence of contradicting data without devaluing the process of reasoning. 
\end{abstract}

\section{Introduction}
K\"{u}rbis \cite{KurbisNegFree19a} developed a theory of definite descriptions formalised with a binary quantifier $ \I $ such that $ \I x[F,G]$ means `the $ F $ is $ G $'. This theory is based on intuitionistic first-order logic with identity and its negative free version. Later on, K\"{u}rbis presented another version based on intuitionistic positive free logic \cite{KurbisPosFree21a}. The version presented in \cite{KurbisNegFree19a} is a Russellian one; $ \I x[F,G]$ is equivalent to Russell's definition of a definite description, that is, $ \exists x(F \wedge \forall y(F^x_y\rightarrow y=x)\wedge G)$. However, Russell does not use a binary quantifier, but a term-forming iota-operator $ \iota $: `the $ F $ is $ G $' in Russell's notation is written as $ G(\iota x F (x)) $. As noticed in \cite{KurbisPosFree21a}, one of the problems with this notation is the meaning of $ \neg G(\iota x F (x)) $: it might be understood as `the $ F $ is not
$ G $' or as `that it is not the case that the $ F $ is $ G $'. The use of a binary quantifier allows K\"{u}rbis to escape from this ambiguity. So `the $ F $ is not
$ G $' is formalised as $ \I x[F,\neg G]$ and `that it is not the case that the $ F $ is $ G $' as $ \neg\I x[F, G]$.
 
Generally speaking, the Russellian method might lead to contradiction. There are several ways to deal with that: require $ G $ in $ G(\iota x F(x)) $ to be atomic, introduce scope distinctions, use free logic, use $ \lambda $-calculus, use paraconsistent logic. In our opinion, the first approach is too restrictive, the second approach might be too clumsy. Free logics lack the deductive strength necessary to deduce a contradiction. Free logic is quite often employed in the study of definite descriptions and is a good solution. The use of $ \lambda $-calculus works fine as well, although makes the language more complicated. We would like to examine the last option, the use of paraconsistent logic, which is a rather rarely explored option.  Contradiction ceases to be an issue in a paraconsistent logic since it prevents us from drawing all the possible conclusions. Therefore, we may answer this problem without employing free logic or $ \lambda $-calculus by using Nelson's logic \textbf{N4} \cite{AlmukdadNelson} as the foundation for the research of $ \I $.
 
 Intuitionistic logic is known for its constructive view of truth. Nelson's logic \textbf{N4} \cite{AlmukdadNelson} (as well as its non-paraconsistent version \textbf{N3} \cite{Nelson}) makes falsity constructive too. One of the aims of this paper is to formulate K\"{u}rbis' approach to definite descriptions on the basis of logic with both truth and falsity being constructive. So we study $ \I $ in Nelson's \textbf{N4}-first order logic and in its negative free version.
 
 
 To sum up, our motivation is to avoid negative consequences of contradictions in Russellian theory of definite descriptions by the use of paraconsistent logic and to make this theory constructive, in such a way that both truth and falsity are constructive. 
 The choice of \textbf{N4} allows to reach both aims.
 
 K\"{u}rbis' \cite{KurbisNegFree19a} approach is proof-theoretic: he uses Tennant's \cite{Tennant78} natural deduction system for intuitionistic first-order logic with identity as well as Tennant's natural deduction system for intuitionistic negative free logic with identity and extends them by the rules for $ \I $.\footnote{Actually, Tennant has his own approach to definite descriptions \cite{Tennant78,Tennant04} and the rules for $ \iota $; the paper \cite{KurbisNegFree19b} compares K\"{u}rbis' and Tennant's methods.} In keeping with this, we also present our results in the form of natural deduction systems. But unlike K\"{u}rbis, we also use semantics in our work. Additionally, we establish the following embedding theorems: both syntactically and semantically Nelson's (negative free) logic is embedded into intuitionistic (negative free) logic. As a consequence, we obtain the completeness theorem. Instead of using our embedding processes for $ \I $, we utilise its definition via quantifiers to derive the sufficient truth and falsity conditions for $ \I $.
 
The structure of the paper is as follows. In Section \ref{ND}, we formulate natural deduction systems for the logics in question. In Section \ref{SM}, we formulate the semantics for these natural deduction systems. In Section \ref{EMB}, we formulate an embedding function and prove embedding theorems. Section \ref{CON} makes concluding remarks.
\section{Natural deduction calculi}\label{ND}
Let us fix a first-order language $ L^\neg $ with the following symbols: variables $ v_1,v_2,\ldots $; constants: $ k, k_1, \ldots $; for every natural number $ n > 0 $, $ n $-place predicate letters $ P_0,P_1,P_2,\ldots $; identity predicate $ = $; propositional connectives $ \neg,\wedge,\vee,\rightarrow $; quantifiers: $ \forall,\exists $; comma, left and right parenthesis. In the case of free logic, we use the symbol $ \Ex $ for the existence predicate. In the metalanguage, we write $ x, y, z $ for arbitrary variables, $ a, b, c $ for arbitrary constants, $ t,t_1,t_2,\ldots $ for terms, $ A,B,C,F,G $ for formulas. The notions of a term and a formula of the language $ L^\neg $ are defined in a standard way. Let $ L^\neg_\I $ be an extension of $ L^\neg $ by a binary quantifier $ \I $. Let $ L^\bot $ ($ L^\bot_\I $) be the language obtained from $ L^\neg $ ($ L^\neg_\I $) by the replacement $ \neg $ with constant falsum $ \bot $. Following K\"{u}rbis \cite{KurbisNegFree19a}, we use the following notation:
\begin{quote}
``I will use $ A^x_t $
to denote the result of replacing all
free occurrences of the variable $ x $ in the formula $ A $ by the term $ t $ or the
result of substituting $ t $ for the free variable $ x $ in $ A $. $ t $ is free for $ x $ in $ A $
means that no (free) occurrences of a variable in $ t $ become bound by a
quantifier in $ A $ after substitution. In using the notation $ A^x_t $
I assume that
$ t $ is free for $ x $ in $ A $ or that the bound variables of $ A $ have been renamed to
allow for substitution without `clashes' of variables, but for clarity I also
often mention the condition that $ t $ is free for $ x $ in $ A $ explicitly. I also use
the notation $ Ax $ to indicate that $ x $ is free in $ A $, and $ At $ for the result of
substituting $ t $ for $ x $ in $ A $.'' \cite[p. 82]{KurbisNegFree19a}
\end{quote}

In what follows, we write \textbf{N4} for a first-order version with identity of Nelson's paraconsistent logic from \cite{AlmukdadNelson}, and $ \bf N4^{NF}$ for its negative free version; their extensions by $ \I $ we denote as $ \bf N4_\I$ and $ \bf N4^{NF}_\I$. We write \textbf{Int} for first-order intuitionsitic logic with identity, and $ \bf Int^{NF}$ for its negative free version; similarly, $ \bf Int_\I$ and $ \bf Int^{NF}_\I$ are extensions of $ \bf Int$ and $ \bf Int^{NF}$ by $ \I $.

Based on Prawitz's research \cite{Prawitz} as well as K\"{u}rbis' investigation \cite{KurbisNegFree19a} of the rules for $ \I $, we formulate the following Gentzen-Prawitz-style  natural deduction systems for $ \bf N4$, $ \bf N4^{NF}$, $ \bf N4_\I$, and $ \bf N4^{NF}_\I$. The difference between free and non-free logics lies in the rules for quantifiers, including $ \I $, identity (the existence predicate $ \Ex $ is used in the case of free logics), and the usage of special rules for predicates in the case of free logics. 

The rules for non-negated propositional connectives are as follows:
 \begin{center}
  ($\vee I_{1} $) $ \dfrac{ A}{ A \vee  B} $ \quad	
  ($\vee I_{2} $) $ \dfrac{ B}{ A \vee  B} $ \quad
 $(\vee E)^{i,j}$ \, $\dfrac{\begin{matrix} 	
  & [A]^i & [B]^j \\
 & \De_1 & \De_2 \\ 	
  A \vee  B &   C &   C \\\end{matrix}}	
{C}
 $  \quad 	
 $ (\rightarrow I)^i$ \, $\dfrac{\begin{matrix} 	
 [A]^i  \\	
 \De \\
   B  \\\end{matrix}}	
 {A \rightarrow  B}
  $  \quad 
\end{center}
\begin{center}
  ($ \rightarrow E$) $ \dfrac{A\rightarrow B\quad A}{B} $\quad
  ($\wedge I$) $ \dfrac{A \quad B}{A \wedge B} $ \quad			
  ($\wedge E_{1} $) $ \dfrac{ A \wedge  B}{ A} $ \quad	
  ($\wedge E_{2} $) $ \dfrac{ A \wedge  B}{ B} $\quad		 
 \end{center}
 
 The rules for negated propositional connectives as follows:
\begin{center}
 ($\neg \neg I$) $ \dfrac{A}{\neg \neg A} $ \quad
 ($\neg \neg E$) $ \dfrac{\neg \neg A}{A} $ \quad
 ($ \neg{\rightarrow} I $) $ \dfrac{A\quad\neg B}{\neg(A\rightarrow B)} $\quad
  ($ \neg{\rightarrow} E_1 $) $ \dfrac{\neg(A\rightarrow B)}{A} $\quad
  ($ \neg{\rightarrow} E_2 $) $ 
  \dfrac{\neg(A\rightarrow B)}{\neg B} $\quad 
   \end{center} 
   \begin{center}   
($\neg\!\vee\! I$) $ \dfrac{\neg  A\quad \neg  B}{\neg ( A \vee  B)} $ \quad			
($\neg\! \vee\! E_1$) $ \dfrac{\neg ( A \vee  B)}{\neg  A} $ \quad
($\neg\! \vee\! E_2$) $ \dfrac{\neg ( A \vee  B)}{\neg  B} $ \quad
\end{center} 
\begin{center} 
($\neg\!\wedge\! I_1$) $ \dfrac{\neg  A}{\neg ( A \wedge  B)} $ \quad	
($\neg\!\wedge\! I_2$) $ \dfrac{\neg  B}{\neg ( A \wedge  B)} $ \quad	
$(\neg\! \wedge\! E)^{i,j}$ \, $\dfrac{\begin{matrix} 	
  & [\neg A]^i & [\neg B]^j \\
 & \De_1 & \De_2 \\ 	
\neg ( A \wedge  B) &   C &   C \\\end{matrix}}	
{C}
 $  \quad 
\end{center}

The rules for quantifiers are as follows (we give them in both ordinary and free versions (the rules for an ordinary version contain $ ^\prime $ in their names); the proviso below is given in the form suitable for free version, but can be straightforwardly adapted for the ordinary one):	

\begin{center}
$ (\forall I)^i $ \, $\dfrac{\begin{matrix} 	
[\Ex y]^i  \\	
\De \\
  A^x_y  \\\end{matrix}}	
{\forall x A}
 $  \; 
 $ ( \forall E) $ $ \dfrac{\forall x A\quad \Ex t}{A^x_t} $\;
 $ (\neg\forall I) $ $ \dfrac{\neg A^x_t\quad \Ex t}{\neg\forall x A} $\;
 $ (\neg\forall E)^{i} $ \, $\dfrac{\begin{matrix} 	
   & [\neg A^x_y]^i, [\Ex y]^i  \\
  & \De \\  	
\neg\forall x A &   C  \\\end{matrix}}	
 {C}
  $  \;
\end{center}
\begin{center}
$ (\forall I^\prime) $ \, $\dfrac{\begin{matrix} 		
  A^x_y  \\\end{matrix}}	
{\forall x A}
 $  \; 
 $ ( \forall E^\prime) $ $ \dfrac{\forall x A}{A^x_t} $\;
 $ (\neg\forall I^\prime) $ $ \dfrac{\neg A^x_t}{\neg\forall x A} $\;
 $ (\neg\forall E^\prime)^{i} $ \, $\dfrac{\begin{matrix} 	
   & [\neg A^x_y]^i  \\
  & \De \\  	
\neg\forall x A &   C  \\\end{matrix}}	
 {C}
  $  \;
\end{center}

where in $ (\forall I) $, $ y $ does not occur free in any undischarged assumptions of $ \De $
except $ \Ex y $, and either $ y $ is the same as $ x $ or $ y $ is not free in $ A $; in $ (\forall E) $, $ t $ is free for $ x $ in $ A $; 
in $ (\neg\forall I) $, $ t $ is free for $ x $ in $ A $; and in $ (\neg\forall E) $, $ y $ is not free in $ C $ nor any
undischarged assumptions of $ \De $, except $ \neg A^x_y $ and $ \Ex y $, and either $ y $ is the
same as $ x $ or it is not free in $ A $.

\begin{center}
$ (\exists I) $ $ \dfrac{A^x_t\quad \Ex t}{\exists x A} $\quad
$ (\exists E)^{i} $ \, $\dfrac{\begin{matrix} 	
  & [A^x_y]^i, [\Ex y]^i  \\
 & \De \\  	
\exists x A &   C  \\\end{matrix}}	
{C}
 $  \quad
 $ (\neg\exists I)^i $ \, $\dfrac{\begin{matrix} 	
 [\Ex y]^i  \\	
 \De \\
  \neg A^x_y  \\\end{matrix}}	
 {\neg\exists x A}
  $  \quad 
  $ (\neg\exists E) $ $ \dfrac{\neg\exists x A\quad \Ex t}{\neg A^x_t} $\quad
\end{center}
\begin{center}
$ (\exists I^\prime) $ $ \dfrac{A^x_t}{\exists x A} $\quad
$ (\exists E^\prime)^{i} $ \, $\dfrac{\begin{matrix} 	
  & [A^x_y]^i  \\
 & \De \\  	
\exists x A &   C  \\\end{matrix}}	
{C}
 $  \quad
 $ (\neg\exists I^\prime) $ \, $\dfrac{\begin{matrix} 	
  \neg A^x_y  \\\end{matrix}}	
 {\neg\exists x A}
  $  \quad 
  $ (\neg\exists E^\prime) $ $ \dfrac{\neg\exists x A}{\neg A^x_t} $\quad
\end{center}

where in $ (\exists I) $, $ t $ is free for $ x $ in $ A $; and in $ (\exists E) $, $ y $ is not free in $ C $ nor any
undischarged assumptions of $ \De $, except $ A^x_y $ and $ \Ex y $, and either $ y $ is the
same as $ x $ or it is not free in $ A $;  $ (\neg\exists I) $, $ y $ does not occur free in any undischarged assumptions of $ \De $
except $ \Ex y $, and either $ y $ is the same as $ x $ or $ y $ is not free in $ A $; in $ (\neg\exists E) $, $ t $ is free for $ x $ in $ A $.

The rules for identity are given below (both in the ordinary and free versions), where $ A $ is an atomic formula or its negation (the rule $ (=E) $ is suitable for both ordinary and free versions; while $ (=I^\prime) $ is used in an ordinary version and $ (=I) $ in a free one):
\begin{center}
$ (=I) $ $ \dfrac{\Ex t}{t=t} $\quad
$ (=I^\prime) $ $ \dfrac{}{t=t} $\quad
$ (=E) $ $ \dfrac{t_1=t_2\quad A^x_{t_1}}{A^x_{t_2}} $\quad
\end{center}

The special rules for free logic regarding predicates ($ P $ stands for an arbitrary predicate, including $ = $):
\begin{center}
(PD) $ \dfrac{P(t_1,\ldots,t_n)}{\Ex t_i} $\quad
($ \neg $PD) $ \dfrac{\neg P(t_1,\ldots,t_n)}{\Ex t_i} $\quad
\end{center}


The rules for a binary quantifier representation of definite descriptions (both ordinary and free versions):
\begin{center}
$ (\I I)^{i} $ \, $\dfrac{\begin{matrix} 	
   & &  & [F^x_y]^i [\Ex y]^i \\
  & &  & \De \\ 	
F^x_t &   G^x_t & \Ex t & y=t \\\end{matrix}}	
 {\I x[F,G]}
  $  \quad  
$ (\I I^\prime)^{i} $ \, $\dfrac{\begin{matrix} 	
   & &   [F^x_y]^i  \\
  & &   \De \\ 	
F^x_t &   G^x_t &  y=t \\\end{matrix}}	
 {\I x[F,G]}
  $  \quad 
  \end{center}
  
where $ t $ is free for $ x $ in $ F $ and in $ G $, and $ y $ is different from $ x $, not free in $ t $ and does not occur free in any undischarged assumptions in $ \De $ except $ F^x_y $ and $ \Ex y $.  

  \begin{center}
$ (\neg\I E)^{i,j,k} $ \, $\dfrac{\begin{matrix} 	
   & [\neg F^x_t]^{i} & [\neg G^x_t]^{j} &  [F^x_y]^{k} [\Ex y]^{k} [\neg y=t]^{k} \\
  & \De_1 & \De_2 & \De_3  \\ 	
\neg \Ix[F,G]  & C & C & C   \\\end{matrix}}	
 {C}
  $  \quad    
\end{center}
  \begin{center}
$ (\neg\I E^\prime)^{i,j,k} $ \, $\dfrac{\begin{matrix} 	
   & [\neg F^x_t]^{i} & [\neg G^x_t]^{j} &  [F^x_y]^{k}  [\neg y=t]^{k} \\
  & \De_1 & \De_2  & \De_3  \\ 	
\neg \Ix[F,G]  & C & C & C   \\\end{matrix}}	
 {C}
  $  \quad    
\end{center}

where $ t $ is free for $ x $ in $ F $ and in $ G $, and $ y $ is different from $ x $, not free in $ t $ and does not occur free in any undischarged assumptions in $ \De_4 $ except $ F^x_y $ and $ \Ex y $. Free version:
\begin{center}
{\small $ (\I E_1)^{i} $  $\dfrac{\begin{matrix} 	
   &  [F^x_y]^i [G^x_y]^i [\Ex y]^i \\
  &  \De \\ 	
\Ix[F,G] & C \\\end{matrix}}	
 {C}
  $  \; 
$ (\neg \I I_1) $  $ \dfrac{\neg F^x_y}{\neg\Ix[F,G]} $\;
$ (\neg \I I_2) $  $ \dfrac{\neg G^x_y}{\neg\Ix[F,G]} $\;}
\end{center}

Ordinary version:

\begin{center}
{\small $ (\I E_1^\prime)^{i} $  $\dfrac{\begin{matrix} 	
   &  [F^x_y]^i [G^x_y]^i \\
  &  \De \\ 	
\Ix[F,G] & C \\\end{matrix}}	
 {C}
  $  \; 
$ (\neg \I I_1) $  $ \dfrac{\neg F^x_y}{\neg\Ix[F,G]} $\;
$ (\neg \I I_2) $  $ \dfrac{\neg G^x_y}{\neg\Ix[F,G]} $\;}
\end{center}

where is $ y $ not free in $ C $ nor any undischarged assumptions it depends on
except $ F^x_y $, $ G^x_y $, and $ \Ex y $, and either $ y $ is the same as $ x $ or it is not free in $ F $ nor in $ G $.

\begin{center}
{\small $ (\I E_2) $ $ \dfrac{\Ix[F,G]\quad \Ex t_1 \quad \Ex t_2 \quad F^x_{t_1}\quad F^x_{t_2}}{t_1=t_2} $\quad
$ (\neg\I I_3) $ $ \dfrac{\neg t_1=t_2\quad \Ex t_1\quad \Ex t_2 \quad F^x_{t_1}\quad F^x_{t_2}}{\neg\Ix[F,G]} $\quad}
\end{center}
\begin{center}
{\small $ (\I E_2^\prime) $ $ \dfrac{\Ix[F,G]\quad F^x_{t_1}\quad F^x_{t_2}}{t_1=t_2} $\quad
$ (\neg\I I_3^\prime) $ $ \dfrac{\neg t_1=t_2\quad F^x_{t_1}\quad F^x_{t_2}}{\neg\Ix[F,G]} $\quad}
\end{center}
where $ t_1 $ and $ t_2 $ are free for $ x $ in $ F $.



Natural deduction systems for $ \bf Int$, $ \bf Int^{NF}$, $ \bf Int_\I$, and $ \bf Int^{NF}_\I$ can be obtained from natural deduction systems for  $ \bf N4$, $ \bf N4^{NF}$, $ \bf N4_\I$, and $ \bf N4^{NF}_\I$ by implementing the following changes: in the rule ($ =E $), $ A $ stands just for atomics formulas (not their negations), all negated rules for connectives, quantifiers, including $ \I $, and predicates have to replaced with the following rule

\begin{center}
$ (\bot E) $ $ \dfrac{\bot}{B} $
\end{center}

As follows from \cite[p. 85]{KurbisNegFree19a}, $ \I x[F,G] $ and $ \exists x(F \wedge \forall y(F^x_y\rightarrow y=x)\wedge G)$ are interderivable in intuitionstic logic. Since in this proof only non-negated rules are used, it is a proof in Nelson logic as well. Thus, $ \I x[F,G] $ and $ \exists x(F \wedge \forall y(F^x_y\rightarrow y=x)\wedge G)$ are interderivable in Nelson's logic as well. As follows from \cite[p. 90--91]{KurbisNegFree19a}, $ \I x[F,G] $ and $ \exists x(F \wedge \forall y(F^x_y\rightarrow y=x)\wedge G)$ are interderivable in intuitionstic negative free logic as well. Again, the same proof can be used in the case of Nelson's logic, since only non-negated rules are involved, so we can conclude that $ \I x[F,G] $ and $ \exists x(F \wedge \forall y(F^x_y\rightarrow y=x)\wedge G)$ are interderivable in Nelson's free logic.

However, in the case of Nelson's logic a natural question arises: what about negation of $ \I x[F,G] $? We can show that $ \neg\I x[F,G] $ and $ \forall x(\neg F \vee \exists y (F^x_y \wedge \neg y=x)\vee \neg G) $ are interderivable in Nelson's logic. Let us denote $ \neg F \vee \exists y (F^x_y \wedge \neg y=x)\vee \neg G $ via $ \mathfrak{F} $.

1. $\neg\I x[F,G] \vdash_{\bf N4} \forall x(\neg F \vee \exists y (F^x_y \wedge \neg y=x)\vee \neg G)$ (where double line means a double application of a disjunction introduction rule):
\begin{center}
\EnableBpAbbreviations
\AXC{$ \neg\I x[F,G] $}
\AXC{$ [\neg F]^1$}\doubleLine
\UIC{$ \mathfrak{F} $}
\AXC{$ [F^x_y]^2 $}
\AXC{$ [\neg y=x]^3 $}
\RL{$ (\wedge I) $}
\BIC{$ F^x_y \wedge \neg y=x $}
\RL{$ (\exists I) $}
\UIC{$ \exists y (F^x_y \wedge \neg y=x)$}\doubleLine
\UIC{$ \mathfrak{F} $}
\AXC{$ [\neg G]^4$}\doubleLine
\UIC{$ \mathfrak{F} $}
\RL{$(\neg\I E) ^{1,2,3,4} $}
\QuaternaryInfC{$ \mathfrak{F} $}
\RL{$ (\forall I^\prime) $}
\UIC{$ \forall x\mathfrak{F}$}
\DisplayProof
\end{center}

2. $\forall x(\neg F \vee \exists y (F^x_y \wedge \neg y=x)\vee \neg G)\vdash_{\bf N4}\neg\I x[F,G]$. 
 Let us denote $ F^x_y \wedge\neg y=x $ via $ \mathfrak{G}^x_y $.
\begin{center}
{\scriptsize 
\EnableBpAbbreviations
\AXC{$ \forall x\mathfrak{F} $}
\RL{$ (\forall E^\prime) $}
\UIC{$ \mathfrak{F} $}
\AXC{$ [\neg F]^1 $}\doubleLine
\UIC{$ \neg\I x[F,G] $}
\AXC{$ [\exists y (\mathfrak{G})]^2 $}
\AXC{$ [\exists y (\mathfrak{G})]^2 $}
\AXC{$ [\mathfrak{G}^x_y]^3 $}
\UIC{$ \neg y=x$}
\AXC{$ [\mathfrak{G}^x_x]^4 $}
\UIC{$F^x_x$}
\AXC{$ [\mathfrak{G}^x_y]^3 $}
\RL{$ (\wedge E) $}
\UIC{$F^x_y$}
\RL{$ (\neg\I I_3) $}
\TIC{$ \neg\I x[F,G] $}
\RL{$ (\exists E^\prime)^{4} $}
\BIC{$ \neg\I x[F,G] $}
\RL{$ (\exists E^\prime)^{3} $}
\BIC{$ \neg\I x[F,G] $}
\AXC{$ [\neg G]^7 $}\doubleLine
\UIC{$ \neg\I x[F,G] $}
\LL{$(\vee E) ^{1,2,7} $}
\QuaternaryInfC{$ \neg\I x[F,G] $}
\DisplayProof}
\end{center}

In the case of Nelson's free logic we have the following deductions. 

1. $\neg\I x[F,G]\vdash_{\bf N4}\forall x(\neg F \vee \exists y (F^x_y \wedge \neg y=x)\vee \neg G)$.

\begin{center}
\EnableBpAbbreviations
\AXC{$ \neg\I x[F,G] $}
\AXC{$ [\neg F]^1$}\doubleLine
\UIC{$ \mathfrak{F} $}
\AXC{$ [F^x_y]^2 $}
\AXC{$ [\neg y=x]^3 $}
\RL{$ (\wedge I) $}
\BIC{$ F^x_y \wedge \neg y=x $}
\AXC{$ [\Ex y]^5$}
\RL{$ (\exists I) $}
\BIC{$ \exists y (F^x_y \wedge \neg y=x)$}\doubleLine
\UIC{$ \mathfrak{F} $}
\AXC{$ [\neg G]^4$}\doubleLine
\UIC{$ \mathfrak{F} $}
\RL{$(\neg \I E) ^{1,2,3,4} $}
\QuaternaryInfC{$ \mathfrak{F} $}
\RL{$ (\forall I)^{5} $}
\UIC{$ \forall x\mathfrak{F}$}
\DisplayProof
\end{center}

2. $\Ex y,\forall x(\neg F \vee \exists y (F^x_y \wedge \neg y=x)\vee \neg G)\vdash_{\bf N4}\neg\I x[F,G]$. Let us denote $ F^x_y \wedge\neg y=x $ via $ \mathfrak{G}^x_y $.
\begin{center}
{\tiny 
\EnableBpAbbreviations
\AXC{$ \forall x\mathfrak{F} $}
\AXC{$ \Ex y$}
\BIC{$ \mathfrak{F} $}
\AXC{$ [\neg F]^1 $}
\UIC{$ \neg\I x[F,G] $}
\AXC{$ [\exists y (\mathfrak{G})]^2 $}
\AXC{$ [\exists y (\mathfrak{G})]^2 $}
\AXC{$ [\mathfrak{G}^x_y]^3 $}
\UIC{$ \neg y=x$}
\AXC{$ [\Ex x]^4 $}
\AXC{$ [\Ex y]^5 $}
\AXC{$ [\mathfrak{G}^x_x]^6 $}
\UIC{$F^x_x$}
\AXC{$ [\mathfrak{G}^x_y]^3 $}
\UIC{$F^x_y$}
\RL{$ (\neg\I I_3) $}
\QuinaryInfC{$ \neg\I x[F,G] $}
\RL{$ ^{4,6} $}
\BIC{$ \neg\I x[F,G] $}
\RL{$ ^{3,5} $}
\BIC{$ \neg\I x[F,G] $}
\AXC{$ [\neg G]^7 $}
\UIC{$ \neg\I x[F,G] $}
\LL{$ ^{1,2,7} $}
\QuaternaryInfC{$ \neg\I x[F,G] $}
\DisplayProof}
\end{center}

\section{Semantics}\label{SM}
Let us describe semantics for intuitionistic negative free logic with identity as well as intuitionistic first-order logic with identity.
We follow Priest's \cite{Priest} presentation of semantics for intuitionistic first-order logic with identity. 

\begin{defn}[Intuitionisitic negative free structure]
An intuitionistic negative free structure $ \mathfrak{I} $ is a seventuple $ \langle W,R,H,D,E,J, \varphi\rangle  $, where $ W $ is the non-empty set of possible worlds, $ R $ is a binary reflexive and transitive relation on $ W $,  $ H $ is a non-empty set of objects, $ D $ is the non-empty
domain of quantification, which members are functions from $ W $ to $ H $ such that for any $ d\in D$ and $ w\in W$ we have $ d(w)\in H$ (in what follows, we write $ |d|_w $ for $ d(w) $), $ E$ is the (possibly, empty) set of all existent objects such that $ E\subseteq D$, $ J =\{|d|_w \in H\mid d\in E \}$,
$ \varphi $ is a function such that it maps $ w\in W$ to a subset of $ D $, $ \varphi(w)\subseteq D$, which we denote as $ D_w $, and satisfies the following conditions, for any $ w\in W$:
\begin{itemize}
\item $ \varphi_w(\Ex)=J $,
\item if $ c $ is a constant, then $ \varphi(c) \in D_w $,
\item if $ P $ is an $ n $-place predicate, then $  \varphi_w(P) \subseteq J^n $,
\item $ \varphi_w(=)=\{\langle t,t\rangle \mid t\in J\} $,
\item if $ wRw^\prime $, then $ \varphi_w(P)\subseteq\varphi_{w^\prime}(P) $, for any $ n $-place predicate predicate $ P $, including $ = $,
\item if $ wRw^\prime $, then $ D_w\subseteq D_{w^\prime} $.
\item if $ \langle d_1,\ldots,d_n\rangle\in\varphi_w(P) $, then $ d_1\in\varphi_w(\Ex),\ldots,d_n\in\varphi_w(\Ex)$.
\end{itemize}
\end{defn}

\begin{defn}[Intuitionistic structure]
An intuitionistic structure is an intuitionistic negative free structure $ \mathfrak{I} = \langle W,R,H,D,E,J, \varphi\rangle  $ such that $ D=E $, and hence $ H=J $; and $ \varphi_w(\Ex)=D $.
\end{defn}


Following Priest \cite{Priest}, for all $ d \in D $, we add a constant to the language, $ k_d $, such that $ \varphi(k_d) = d $.
\begin{defn}[Intuitionistic (negative free) semantics]\label{IntSemantic}
An intuitionistic (ne\-ga\-ti\-ve free) valuation $ \Vdash^I $ on a model $ \mathfrak{I} = \langle W,R,H,D,E,J, \varphi\rangle  $ is defined as follows, for any $ w\in W$:
\begin{itemize}\itemsep=0pt
\item $ \mathfrak{I},w\Vdash^I P(t_1,\ldots,t_n) $ iff $ \langle |\varphi(t_1)|_w,\ldots,|\varphi(t_n)|_w\rangle\in \varphi_w(P^n) $,
\item $ \mathfrak{I},w\nVdash^I \bot$,
\item $ \mathfrak{I},w\Vdash^I A\rightarrow B$ iff $ \forall w^\prime\in W(R(w,w^\prime)\text{~implies~}(\mathfrak{I},w^\prime\Vdash^I A\text{~implies~}\mathfrak{I},w^\prime\Vdash^I B))$,
\item $ \mathfrak{I},w\Vdash^I A\wedge B$ iff $ \mathfrak{I},w\Vdash^I A $ and $ \mathfrak{I},w\Vdash^I B$,
\item $ \mathfrak{I},w\Vdash^I A\vee B$ iff $ \mathfrak{I},w\Vdash^I A $ or $ \mathfrak{I},w\Vdash^I B$,
\item $ \mathfrak{I},w\Vdash^I\forall x A $ iff $ \forall w^\prime(R(w,w^\prime)\text{~implies~} \forall d\in E_{w^\prime}, \mathfrak{I},w^\prime\Vdash^IA^x_{k_d}) $
\item $ \mathfrak{I},w\Vdash^I\exists x A $ iff $ \exists d\in E_{w}, \mathfrak{I},w\Vdash^IA^x_{k_d} $.
\end{itemize}
Using the fact that $ \I x[F,G] $ and $ \exists x(F \wedge \forall y(F^x_y\rightarrow y=x)\wedge G)$ are interderivable, we can propose the following semantic condition for $ \I x[F,G] $:
\begin{itemize}
\item $ \mathfrak{I},w\Vdash^I\I x[F,G] $ iff $ \exists d\in E_{w}, \mathfrak{I},w\Vdash^IF $ and $\forall w^\prime(R(w,w^\prime)\text{~implies~} \forall e\in E_{w^\prime},  \forall w^{\prime\prime}\in W(R(w^\prime,w^{\prime\prime})$  $\text{~implies~}(\mathfrak{I},w^{\prime\prime}\Vdash^I F^{k_d}_{k_e}\text{~implies~}\mathfrak{I},w^{\prime\prime}\Vdash^I k_d=k_e)))$ and $ \mathfrak{I},w\Vdash^IG $.
\end{itemize}
\end{defn}

The semantics for $ \bf Int$ and $ \bf Int_\I$ is based on intuitionistic structures, and for $ \bf Int^{NF}$ and $ \bf Int_\I^{NF}$ on intuitionistic negative free structures.
\begin{defn}
An inference is valid iff it is truth-preserving in all worlds of all
interpretations.
\end{defn}

Let us present semantics for Nelson's logics on the basis of Thomason's semantics \cite{Thomason} (see also \cite{Priest}). However, in contrast to \cite{Thomason,Priest}, the semantics we use is two-valued with a paradefinite valuation (thus, a formula and its negation can simultaneously be true and false, or simultaneously neither true, nor false).
\begin{defn}[Nelsonian negative free structure]
A Nelsonian ne\-ga\-ti\-ve 
free structure $ \mathfrak{N} $ is an intuitionistic negative free structure $  \langle W,R,H,D,E,J, \varphi\rangle  $ such that $ \varphi $ is redefined as follows:
\begin{itemize}
\item $ \varphi_w(\Ex)=J $, $ \varphi_w(\neg\Ex)=H\setminus J $,
\item if $ c $ is a constant, then $ \varphi(c) \in D $,
\item if $ P $ is an $ n $-place predicate, then $  \varphi_w(P) \subseteq J^n $ and $  \varphi_w(\neg P) \subseteq J^n $,
\item $ \varphi_w(=)=\{\langle t,t\rangle \mid t\in J \} $, $ \varphi_w(\neg =)\subseteq J^2 $,
\item if $ wRw^\prime $, then $ \varphi_w(P)\subseteq\varphi_{w^\prime}(P) $ and $ \varphi_w(\neg P)\subseteq\varphi_{w^\prime}(\neg P) $,
\item if $ wRw^\prime $, then $ D_w\subseteq D_{w^\prime} $,
\item if $ \langle d_1,\ldots,d_n\rangle\in\varphi_w(P) $, then $ d_1\in\varphi_w(\Ex),\ldots,d_n\in\varphi_w(\Ex)$,
\item if $ \langle d_1,\ldots,d_n\rangle\in\varphi_w(\neg P) $, then $ d_1\in\varphi_w(\Ex),\ldots,d_n\in\varphi_w(\Ex)$.
\end{itemize}
\end{defn}

\begin{defn}[Nelsonian structure]
A Nelsonian structure is a Nelsonian negative free structure $ \mathfrak{I} = \langle W,R,H,D,E,J, \varphi\rangle  $ such that $ D=E $, and hence $ H=J $; and $ \varphi_w(\Ex)=\varphi_w(\neg\Ex)=D_w $.
\end{defn}

\begin{defn}[Nelsonian semantics]\label{NelsonSemantic}
A Nelsonian paradefinite valuation $ \Vdash^N $ on a model $ \mathfrak{N} =\linebreak \langle W,R,H,D,E,J, \varphi\rangle$ is defined as follows, for any $ w\in W$:\footnote{The truth conditions for non-negated formulas, including $ \I x[F,G] $, are the same as in the intuitionistic case.}
\begin{itemize}\itemsep=0pt
\item $ \mathfrak{N},w\Vdash^N P(t_1,\ldots,t_n) $ iff $ \langle |\varphi(t_1)|_w,\ldots,|\varphi(t_n)|_w\rangle\in \varphi_w(P^n) $,
\item $ \mathfrak{N},w\Vdash^N \neg P(t_1,\ldots,t_n) $ iff $ \langle |\varphi(t_1)|_w,\ldots,|\varphi(t_n)|_w\rangle\in \varphi_w(\neg P^n) $,
\item $ \mathfrak{N},w\Vdash^N\neg\neg A$ iff $ \mathfrak{N},w\Vdash^NA$,
\item $ \mathfrak{N},w\Vdash^N A\rightarrow B$ iff $ \forall w^\prime\in W(R(w,w^\prime)\text{~implies~}(\mathfrak{N},w^\prime\Vdash^N A\text{~implies~}\mathfrak{N},w^\prime\Vdash^N B))$,
\item $ \mathfrak{N},w \Vdash^N \neg(A\rightarrow B)$ iff $ \mathfrak{N},w\Vdash^N A $ and $ \mathfrak{N},w\Vdash^N\neg B$,
\item $ \mathfrak{N},w\Vdash^N A\wedge B$ iff $ \mathfrak{N},w\Vdash^N A $ and $ \mathfrak{N},w\Vdash^N B$,
\item $ \mathfrak{N},w\Vdash^N \neg(A\wedge B)$ iff $ \mathfrak{N},w\Vdash^N\neg A $ or $ \mathfrak{N},w\Vdash^N\neg B$,
\item $ \mathfrak{N},w\Vdash^N A\vee B$ iff $ \mathfrak{N},w\Vdash^N A $ or $ \mathfrak{N},w\Vdash^N B$,
\item $ \mathfrak{N},w\Vdash^N \neg(A\vee B)$ iff $ \mathfrak{N},w\Vdash^N\neg A $ and $ \mathfrak{N},w\Vdash^N\neg B$,
\item $ \mathfrak{N},w\Vdash^N\forall x A $ iff $ \forall w^\prime(R(w,w^\prime)\text{~implies~} \forall d\in D_{w^\prime}, \mathfrak{N},w^\prime\Vdash^NA^x_{k_d}) $,
\item $ \mathfrak{N},w\Vdash^N\neg\forall x A $ iff $ \exists d\in D_{w^\prime}, \mathfrak{N},w^\prime\Vdash^N\neg A^x_{k_d} $,
\item $ \mathfrak{N},w\Vdash^N\exists x A $ iff $ \exists d\in D_{w^\prime}, \mathfrak{N},w^\prime\Vdash^NA^x_{k_d} $,
\item $ \mathfrak{N},w\Vdash^N\neg\exists x A $ iff $ \forall w^\prime(R(w,w^\prime)\text{~implies~} \forall d\in D_{w^\prime}, \mathfrak{N},w^\prime\Vdash^N\neg A^x_{k_d}) $;
\end{itemize}

Using the fact that $ \I x[F,G] $ and $ \exists x(F \wedge \forall y(F^x_y\rightarrow y=x)\wedge G)$ are interderivable as well as $ \neg\I x[F,G] $ and $\Ex y,\forall x(\neg F \vee \exists y (F^x_y \wedge \neg y=x)\vee \neg G)$ are interderivable, we can propose the following semantic condition for $ \I x[F,G] $ and $ \neg\I x[F,G] $:
\begin{itemize}
\item $ \mathfrak{N},w\Vdash^N\I x[F,G] $ iff $ \exists d\in E_{w}, \mathfrak{N},w\Vdash^IF $ and $\forall w^\prime(R(w,w^\prime)\text{~implies~} \forall e\in E_{w^\prime},  \forall w^{\prime\prime}\in W(R(w^\prime,w^{\prime\prime})$  $\text{~implies~}(\mathfrak{N},w^{\prime\prime}\Vdash^N F^{k_d}_{k_e}\text{~implies~}\mathfrak{N},w^{\prime\prime}\Vdash^N k_d=k_e)))$ and $ \mathfrak{N},w\Vdash^NG $,
\item $ \mathfrak{N},w\Vdash^N\neg\I x[F,G] $ iff $ \langle |\varphi(y)|_w\rangle\in \varphi_w(\Ex) $ and $ \forall w^\prime(R(w,w^\prime)\text{~implies~} \forall d\in D_{w^\prime}, \mathfrak{N},w^\prime\Vdash^N\neg F\text{~or~}\exists e\in D_{w^\prime}, (\mathfrak{N},w^\prime\Vdash^NF^{k_d}_{k_e}\text{~and~}\mathfrak{N},w^\prime\Vdash^N \neg k_e=k_d)$ or $\mathfrak{N},w^\prime\Vdash^N\neg G) $.
\end{itemize}
\end{defn}

The semantics for $ \bf N4$ and $ \bf N4_\I$ is based on intuitionistic structures, and for $ \bf N4^{NF}$ and $ \bf N4_\I^{NF}$ on intuitionistic negative free structures.
\begin{defn}
An inference is valid iff it is truth-preserving in all worlds of all
interpretations.
\end{defn}
\section{Embedding theorems}\label{EMB}
We use an embedding function similar to the one used by Gurevich \cite{Gurevich}, Rautenberg \cite{Rautenberg}, Vorob'ev \cite{Vorobev} for \textbf{N3} and \textbf{Int} as well as Kamide and Shramko \cite{KamideShramko} for some multilattice logics. One of the specifics this function is the necessity to extend the language of intuitionistic logic with the additional copies of predicate letters. So extend the language $ L^\bot $ with the set $ \{P^\prime\mid P\text{~is a predicate letter} \} $.

\begin{defn}\label{tau}
An embedding function $ \tau $ from the language $ L^\non $ into the language $ L^\bot $ is inductively defined as follows\textup{:}
	\begin{enumerate}[$(1)$]\itemsep=0pt
\item $ \tau(P(t_1,\ldots,t_n))=P(t_1,\ldots,t_n) $,  for any predicate $ P $,
\item  $ \tau(\neg P(t_1,\ldots,t_n))=P^\prime(t_1,\ldots,t_n) $, for any predicate $ P $, 
		\item $ \tau(A\ast B)=\tau(A)\ast \tau(B) $, where $ \ast\in\{\rightarrow,\wedge,\vee \}$
		\item $ \tau(\non\non A)= \tau(A)$,
		\item $ \tau(\non(A\rightarrow B))= \tau(A) \wedge \tau(\non B) $,
		\item $ \tau(\non(A\wedge B))=\tau(\non A)\vee \tau(\non B) $,
		\item $ \tau(\non(A\vee B))=\tau(\non A)\wedge \tau(\non B) $,
\item $ \tau(\forall x A)=\forall x\tau(A) $,		
\item $ \tau(\exists x A)=\exists x\tau(A) $,
\item $ \tau(\neg\forall x A)=\exists x\tau(\neg A) $,		
\item $ \tau(\neg\exists x A)=\forall x\tau(\neg A) $.
	\end{enumerate}
\end{defn} 

Let us prove the following theorem for $ \bf N4  $ and $ \bf Int $ as well as their negation free versions. A similar theorem has been proven in \cite{Gurevich,Rautenberg,Vorobev} for $ \bf N3$ and $ \bf Int $.
\begin{thm}[Syntactical embedding]\label{SyntacticEmb}
Let $ \tau $ be a mapping introduced in Definition \ref{tau}. For any formula $ A $, $ \vdash_{\bf N4}A $ iff $ \vdash_{\bf Int}\tau(A) $; $ \vdash_{\bf N4^{NF}}A $ iff $ \vdash_{\bf Int^{NF}}\tau(A) $.
\end{thm}
\begin{proof}
As an example, we present a proof for the case of negative free logics.
 
Suppose that $ \vdash_{\bf N4^{NF}}A $. By an induction on the length of the deduction
of $ A $. We distinguish cases depending on the last rule applied in the
deduction. 

Suppose that $ A $ is of the form $ \Ex t_i $ and has been obtained by the rule ($ \non $PD) from the formula $\non P(t_1,\ldots,t_n) $. By the induction hypothesis, the translation $\tau(\Ex t_i)$ is provable in $ \bf Int$. Then we can construct a deduction of the translation of $ \tau(\Ex t_i)=\Ex t_i $ in \textbf{Int} using the rule (PD):
\begin{center}
	\EnableBpAbbreviations
	\AXC{$\non P(t_1,\ldots,t_n)$}
	\RL{($\non$PD)}
	\UIC{$\Ex t_i$}
	\DisplayProof$\; \rightsquigarrow\; $
	\EnableBpAbbreviations
	\AXC{$P^\prime(t_1,\ldots,t_n)$}
	\RL{(PD)}
	\UIC{$\Ex t_i$}
	\DisplayProof
\end{center}

Suppose that $ A $ is of the form $ \non(B\rightarrow C) $ and has been obtained by the rule $ (\non{\rightarrow} I) $ from the formulas $ B $ and $ \non C $. By the induction hypothesis, the translations $\tau(B)$ and $\tau(\non C)$ are provable in $ \bf Int$. Then we can construct a deduction of the translation of $ \tau(\non(B\rightarrow C))=\tau(B)\wedge \tau(\non C) $ in \textbf{Int} using the rule $(\wedge I )$:
\begin{center}
	\EnableBpAbbreviations
	\AXC{$B$}
	\AXC{$\non C$}
	\RL{$ (\non{\rightarrow} I) $}
	\BIC{$ \non(B\rightarrow C)$}
	\DisplayProof$\; \rightsquigarrow\; $
	\EnableBpAbbreviations
	\AXC{$\tau(B)$}
	\AXC{$\tau(\non C)$}
	\RL{$ (\wedge I) $}
	\BIC{$ \tau(B)\wedge \tau(\non C)$}
	\DisplayProof
\end{center}

Suppose that $ A $ is of the form $ \neg\forall x B $ and has been obtained by the rule $ (\non\forall I) $ from the formulas $\neg B^x_t $ and $ \Ex t$. By the induction hypothesis, the translations $\tau(\neg B^x_t)$ and $\tau(\Ex t)$ are provable in $ \bf Int$. Then we can construct a deduction of the translation of $ \tau(\neg\forall x B)=\exists x\tau(\neg B) $ in \textbf{Int} using the rule $(\exists I )$:
\begin{center}
	\EnableBpAbbreviations
	\AXC{$\neg B^x_t$}
	\AXC{$\Ex t$}
	\RL{$ (\non\forall I) $}
	\BIC{$ \neg\forall x B$}
	\DisplayProof$\; \rightsquigarrow\; $
	\EnableBpAbbreviations
	\AXC{$\tau(\neg B^x_t)$}
	\AXC{$\tau(\Ex t)$}
	\RL{$ (\exists I ) $}
	\BIC{$ \exists x\tau(\neg B)$}
	\DisplayProof
\end{center}

The other cases are considered similarly.

Suppose that $ \vdash_{\bf Int^{NF}}\tau(A) $. Similarly to previous cases. 
\end{proof}

\begin{lemma}\label{SemanticLemma1}
Let $ \mathfrak{N}=\langle W,R,H,D,E,J, \varphi\rangle $ be a Nelsonian (negative free) structure. Let $ \tau $ be the mapping defined in Definition \ref{tau}. For
any Nelsonian paradefinite valuation $ \Vdash^N $ on $ \mathfrak{N} $, we can construct an intuitionistic valuation $ \Vdash^I $ on an
intuitionistic (negative free) structure $ \mathfrak{I}=\langle W,R,H,D,E,J, \varphi\rangle $ such that for any formula $ C $, $ \mathfrak{N}\Vdash^NC $ iff $ \mathfrak{I}\Vdash^I\tau(C) $.
\end{lemma}
\begin{proof}
As an example, we give a proof for the case of non-free logics. 
Let $ \mathcal{P} $ be a set of atomic formulas and let $ \mathcal{P}^\prime $ be the set $ \{P^\prime(t_1,\ldots,t_n) \mid P(t_1,\ldots,t_n)\in\mathcal{P} \} $ of
atomic formulas. Suppose that $ \Vdash^N $ is a Nelsonian paradefinite valuation on $ \mathfrak{N} $. Suppose that $ \Vdash^I $ is
an intuitionistic valuation on $ \mathfrak{I} $ such that, for any $ w\in W $ and for any atomic formula $ P(t_1,\ldots,t_n)\in \mathcal{P}\cup \mathcal{P}^\prime $,
\begin{enumerate}[(a)]\itemsep=0pt
\item $ \mathfrak{N},w\Vdash^N P(t_1,\ldots,t_n) $ iff $ \mathfrak{I},w\Vdash^I P(t_1,\ldots,t_n) $,
\item $ \mathfrak{N},w\Vdash^N\neg P(t_1,\ldots,t_n) $ iff $ \mathfrak{I},w\Vdash^I P^\prime(t_1,\ldots,t_n) $.
\end{enumerate}
The lemma is proved by induction on $ C $.

\begin{enumerate}[(1)]\itemsep=0pt
\item 
 $C$ is an atomic formula $ P(t_1,\ldots,t_n) $: $ \mathfrak{N},w\Vdash^N P(t_1,\ldots,t_n) $ iff $ \mathfrak{I},w\Vdash^I P(t_1,\ldots,t_n) $ (by the
assumption) iff $ \mathfrak{I},w\Vdash^I \tau(P(t_1,\ldots,t_n)) $ (by Definition \ref{tau}).
\item   
 $C$ is a negated atomic formula $ \neg P(t_1,\ldots,t_n) $: $ \mathfrak{N},w\Vdash^N\neg P(t_1,\ldots,t_n) $ iff $ \mathfrak{I},w\Vdash^I P^\prime(t_1,\ldots,t_n) $ (by the
assumption) iff $ \mathfrak{I},w\Vdash^I \tau(\neg P(t_1,\ldots,t_n)) $ (by Definition \ref{tau}).
\item $C$ is $ A\rightarrow B$: $ \mathfrak{N},w\Vdash^N A\rightarrow B$ iff $ \forall w^\prime\in W(R(w,w^\prime)$ implies $(\mathfrak{N},w^\prime\Vdash^N A\text{~implies~}\mathfrak{N},w^\prime\Vdash^N B))$ (by Definition \ref{IntSemantic}) iff 
$ \forall w^\prime\in W(R(w,w^\prime)\text{~implies~}(\mathfrak{I},w^\prime\Vdash^I \tau(A)\text{~implies~}\mathfrak{I},w^\prime\Vdash^I \tau(B)))$ (by the induction hypothesis) iff $ \mathfrak{I},w\Vdash^I \tau(A\rightarrow B)$ (by Definition \ref{IntSemantic}).
\item $C$ is $ \neg(A\rightarrow B)$: $ \mathfrak{N},w \Vdash^N \neg(A\rightarrow B)$ iff $ \mathfrak{N},w\Vdash^N A $ and $ \mathfrak{N},w\Vdash^N\neg B$ (be Definition \ref{NelsonSemantic}) iff  $ \mathfrak{I},w\Vdash^I \tau(A) $ and $ \mathfrak{I},w\Vdash^I\tau(\neg B)$ (by the induction hypothesis) iff $ \mathfrak{I},w \Vdash^I \tau(A)\wedge \tau (\neg B)$ (by Definition \ref{IntSemantic}) iff $ \mathfrak{I},w \Vdash^I \tau(\neg(A\rightarrow B))$ (by Definition \ref{tau}).
\item $ C $ is $ \forall x A$: $ \mathfrak{N},w\Vdash^N\forall x A $ iff $ \forall w^\prime(R(w,w^\prime)\text{~implies~} \forall d\in D_{w^\prime}, \mathfrak{N},w^\prime\Vdash^NA^x_{k_d}) $ (by Definition \ref{NelsonSemantic}) iff $ \forall w^\prime(R(w,w^\prime)\text{~implies~} \forall d\in D_{w^\prime}, \mathfrak{I},w^\prime\Vdash^I\tau(A^x_{k_d})) $ (by the induction hypothesis) iff $ \mathfrak{I},w\Vdash^I\forall x A $ (by Definition \ref{IntSemantic}) iff $ \mathfrak{I},w\Vdash^I\tau(\forall x A) $ (by Definition \ref{tau}).
\item $ C $ is $ \neg\forall x A$: $ \mathfrak{N},w\Vdash^N\neg\forall x A $ iff $ \exists d\in D_{w^\prime}, \mathfrak{N},w^\prime\Vdash^N\neg A^x_{k_d} $ (by Definition \ref{NelsonSemantic}) iff $ \exists d\in D_{w^\prime}, \mathfrak{I},w^\prime\Vdash^I\tau(\neg A^x_{k_d}) $ (by the induction hypothesis) iff $ \mathfrak{I},w\Vdash^I\exists x \tau(\neg A) $ (by Definition \ref{IntSemantic}) iff $ \mathfrak{I},w\Vdash^I\tau(\neg\forall x A) $ (by Definition \ref{tau}).
\end{enumerate}

The other cases are considered similarly.
\end{proof}

\begin{lemma}\label{SemanticLemma2}
Let $ \mathfrak{I}=\langle W,R,H,D,E,J, \varphi\rangle $ be an intuitionistic (negative free) structure. Let $ \tau $ be the mapping defined in Definition \ref{tau}. For
any intuitionistic valuation $ \Vdash^I $ on $ \mathfrak{I} $, we can construct a Nelsonian paraconsistent valuation $ \Vdash^N $ on an
Nelsonian (negative free) structure $ \mathfrak{N}=\langle W,R,H,D,E,J, \varphi\rangle $ such that for any formula $ C $, $ \mathfrak{N}\Vdash^NC $ iff $ \mathfrak{I}\Vdash^I\tau(C) $.
\end{lemma}
\begin{proof}
Similarly to Lemma \ref{SemanticLemma1}. 
\end{proof}
\begin{thm}[Semantic embedding]\label{SemanticEmb}
Let $ \tau $ be a mapping introduced in Definition \ref{tau}. For any formula $ C $, $ \models_{\bf N4}C $ iff $ \models_{\bf Int}\tau(C) $; $ \models_{\bf N4^{NF}}C $ iff $ \models_{\bf Int^{NF}}\tau(C) $.
\end{thm}
\begin{proof}
Follows from Lemmas \ref{SemanticLemma1} and \ref{SemanticLemma2}.
\end{proof}
\begin{thm}[Completeness]\label{Compl}
For any formula $ C $, $ \models_{\bf N4}C $ iff $ \vdash_{\bf N4}C $; $ \models_{\bf N4^{NF}}C $ iff $ \vdash_{\bf N4^{NF}}C $.
\end{thm}
\begin{proof}
Follows from Theorems \ref{SyntacticEmb} and \ref{SemanticEmb} as well as completeness of intuitionistic first-order logics with identity and its negative free version.
\end{proof}
\begin{lemma}\label{SoundI}
All the rules for $ \I $ and $ \neg\I $ are sound.
\end{lemma}
\begin{proof}
Left for the reader.
\end{proof}
\begin{thm}[Completeness]
For any formula $ C $, it holds that $ \models_{\bf N4_\I}C $ iff $ \vdash_{\bf N4_\I}C $; $ \models_{\bf N4^{NF}_\I}C $ iff $ \vdash_{\bf N4^{NF}_\I}C $.
\end{thm}
\begin{proof}
Follows from Theorem \ref{Compl} and the definition of $ \I $ (that is equivalences proved in Section \ref{ND}) as well as Lemma \ref{SoundI}.
\end{proof}
\section{Conclusion}\label{CON}
In this paper, we examined the behaviour of the binary quantifier $ \I $ in Nelson's first-order logic with identity and its negative free version, drawing inspiration from K\"{u}rbis's method of formalising definite descriptions using $ \I $ added to intuitionistic first-order logic with identity and its negative free version. 
The research described in this article can be continued as follows. As a first task for the future, we leave the problem of an adaptation of the embedding function $ \tau $ for the case $ \I $. As a second task, we can propose to find a proof of the normalisation theorem for the natural deduction systems formulated in this article. As a third task, to conduct a similar study, on the basis of \cite{KurbisPosFree21a}, where $ \I $ is characterised by different natural deduction rules and is studied on the basis of intuitionistic positive free logic. As a fourth task, carry out comparable research based on \textbf{N3} instead of \textbf{N4}, or a non-constructive tabular extension of \textbf{N4}/\textbf{N3} by Peirce's law (in the latter case, one can think about embedding such logics into classical first-order (free) logic).

\paragraph{Acknowledgments.} Special thanks go to Nils K\"{u}rbis for useful comments. The author is grateful for the reviewers for their valuable suggestions. This work was   
 funded by the European Union (ERC, ExtenDD, project number: 101054714). 
Views and opinions expressed are however those of the author(s) only and do not necessarily reflect those of the European Union or the European Research Council. Neither the European Union nor the granting authority can be held responsible for them.

\nocite{*}

\end{document}